\documentclass[a4paper,11pt]{article}

\usepackage[utf8]{inputenc}

\usepackage{amsmath}
\usepackage{amssymb}
\usepackage{amsthm}
\usepackage{bigints}
\usepackage{bm}

\usepackage{tikz,graphicx}

\usepackage{float}
\usepackage{subfig}
\usepackage{color,soul}
\sethlcolor{green}

\newtheorem{theorem}{Theorem}

\newtheorem{definition}{Definition}

\usepackage{authblk}

\title{A QUBO Model for\\Gaussian Process Variance Reduction}
\author[]{Lorenzo Bottarelli, Alessandro Farinelli}
\date{}

\affil[]{Department of Computer Science, University of Verona, Italy}

\begin{document}
\newcommand{\TODO}[1]{{\sethlcolor{red} \hl{TODO #1}}}
\newcommand{\dquotes}[1]{``#1''}
\newcommand{\squotes}[1]{`#1'}

\maketitle
\begin{abstract}
Gaussian Processes are used in many applications to model spatial phenomena. Within this context, a key issue is to decide the set of locations where to take measurements so as to obtain a better approximation of the underlying function. Current state of the art techniques select such set to minimize the posterior variance of the Gaussian process. We explore the feasibility of solving this problem by proposing a novel Quadratic Unconstrained Binary Optimization (QUBO) model. In recent years this QUBO formulation has gained increasing attention since it represents the input for the specialized quantum annealer D-Wave machines. Hence, our contribution takes an important first step towards the sampling optimization of Gaussian processes in the context of quantum computation. Results of our empirical evaluation shows that the optimum of the QUBO objective function we derived represents a good solution for the above mentioned problem. In fact we are able to obtain comparable and in some cases better results than the widely used submodular technique.
\end{abstract}

\section{Introduction} \label{sec:Introduction}

Gaussian processes are a widely used tool in machine learning \cite{Rasmussen,Murphy:2012} and provides a statistical distribution together with a way to model an unknown function $f$.
A Gaussian process (GP) defines a prior distribution over functions, which can be converted into a posterior distribution over functions once we have observed some data.

In many spatial analysis, such as environmental monitoring applications, the unknown scalar field of a phenomenon of interest (e.g. the temperature of the environment or the pH value of water in rivers or in lakes \cite{Hitz2014, Singh:2006}) is modeled using a GP. 
In this context it is necessary to choose a set of locations in space in which to measure the specific phenomenon of interest, or similarly it is necessary to chose the displacement positions of fixed sensors \cite{guestrin2005near,krause2008near,Krause-sensor1}. 
However, in both cases the process is usually costly and one wants to select observations that are especially informative with respect to some objective function. A good choice of sampling locations allows to obtain a better approximation of the underling phenomenon.

Research in the context aims at selecting the set of measurements so as to optimize an important sensing quality function for spatial prediction that is represented by the reduction of predictive variance of the Gaussian Process \cite{Krause-sensor2}.
Das and Kempe \cite{das2008algorithms} showed that, in many cases, the variance reduction at any particular location is submodular.
Submodularity is a property of a specific class of set functions. They encode an intuitive diminishing returns property that allows for a  greedy forward-selection algorithm that is widely exploited in sensing optimization literature \cite{Krause-sensor1,Krause07,Powers,krause2014submodular,tzoumas2018resilient}. 

% Another example, besides environmental monitoring applications, where minimizing the predictive variance of a Gaussian process is an important task is represented by the context of emulators for computer experiments \cite{sunga2016}. In this context it is possible to use a Gaussian process to model an emulator for computationally expensive computer experiments. In order to provide a robust model and an accurate approximation of the relationship between simulation output and untried inputs at a reduced computational cost, a small predictive variance of the Gaussian process is an important objective.

A more recent work aims at simultaneously optimize sensing locations by minimizing the posterior variance of a GP through the use of a gradient descent algorithm \cite{BottarelliGD}. However, this technique makes the assumption that the space where observation can be made is continuous and requires an initialization of sampling points to be optimized.

In general, the selection of the optimal set of measurement locations in order to minimize the posterior variance of a GP is NP-hard given the combinatorial nature of the problem and the exponential number of candidate solutions.
This and may problems in artificial intelligence and pattern recognition are computationally difficult due to their inherent complexity and the exponential size of the solution space. 
Quantum information processing could provide a viable alternative to combat such a complexity. 
A notable progress in this direction is represented by the recent development of the D-Wave quantum annealer, whose processor has been designed to the purpose of solving Quadratic Unconstrained Binary Optimization (QUBO) problems.
As a consequence, many works in literature investigate the possibility of using quantum annealing to address hard artificial intelligence and pattern recognition problems by proposing their QUBO formulation.

Examples include image recognition \cite{neven1}, bayesian network structure learning \cite{gorman1}, fault detection and diagnosis \cite{perdomo1}, training a binary classifier \cite{neven2} and portfolio optimization \cite{rosenberg2016solving,marzec2016portfolio,venturelli2018reverse}.
Moreover, in the context of mathematical logic, Bian et al. \cite{bian2018solving} propose a QUBO formulation to tackle the maxSAT problem, an optimization extension of the well known SAT (boolean satisfiability) problem \cite{cook1971complexity}. 

NASA’s Quantum Artificial Intelligence Laboratory (QuAIL) team\footnote{https://ti.arc.nasa.gov/tech/dash/groups/physics/quail/} hosts one of the D-Wave machine and aims to investigate whether quantum computing can improve the ability to address difficult optimization and machine learning problems related to several fields that include NASA's aeronautics, Earth and space sciences, and space exploration missions. 
The focus of the QuAIL team is both theoretical and empirical investigations of quantum annealing. Biswas et al. \cite{BISWAS201781} reviews NASA perspective on quantum computing of three potential application areas such as planning and scheduling \cite{Smelyanskiy1,Rieffel1,venturelli2015quantum,tran2016explorations,tran2016hybrid}, fault detection and diagnosis \cite{perdomo1}, and sampling/machine learning \cite{benedetti2016estimation,amin2015searching,adachi2015application,benedetti2016quantum}.
These works are part of the emerging field of quantum machine learning \cite{schuld2015introduction} where the use of quantum computing technologies for sampling and machine learning applications has attracted increasing attention in recent years.

In this paper we investigate this possibility by proposing a QUBO model to optimize a set of sensing locations. More in details the contributions of this paper are:
\begin{itemize}
\item We propose a QUBO model to minimize the posterior variance of a Gaussian process.
\item We provide a mathematical demonstration that the optimum of our QUBO model satisfies the constraint of the problem.
\item We study the performance of the proposed QUBO model with respect to the submodular greedy algorithm and a random selection.
\end{itemize}

\section{Background} \label{sec:Background}

\subsection{Gaussian Processes} \label{subsec:Gaussian Processes}

A Gaussian Process is a flexible and non-parametric tool that defines a prior distribution over functions, which can be converted into a posterior distribution over functions once we have observed some data. 
A GP is completely defined by its mean and kernel function (also called covariance function) which encodes the smoothness properties of the modeled function $f$. 

Suppose we observe a training set  $\mathcal{K}=\{ (\bm{\mu}_i,y_i) | i=1, \dots, K \}$, that is, a set of $K$ measurements $\{y_1,y_2,\cdots,y_K\}$ taken at locations $\{\bm{\mu}_1,\bm{\mu}_2,\cdots,\bm{\mu}_K\}$.
We consider Gaussian processes that are estimated based on a set of noisy measurements. Hence, we assume that $y_i=f(\textbf{x}_i)+ \epsilon$ where $\epsilon \sim \mathcal{N}(0,\sigma_n^2)$, that is, observations with additive independent identically distributed Gaussian noise $\epsilon$ with variance $\sigma_n^2$. The posterior mean and variance over $f$ for a test point $\textbf{x}_*$ can be computed as follows \cite{Rasmussen,Murphy:2012}:

\begin{equation} \label{eq:mu}
\overline{f}(\textbf{x}_*) = \textbf{k}_*^T (\textbf{K} + \sigma_n^2\textbf{I})^{-1}\textbf{y}
\end{equation}
\begin{equation} \label{eq:sigma^2}
\sigma^2(\textbf{x}_*) = k(\textbf{x}_*,\textbf{x}_*) - \textbf{k}_*^T(\textbf{K} + \sigma_n^2 \textbf{I})^{-1}\textbf{k}_*
\end{equation}
where $\textbf{k}_* = [k(\bm{\mu}_1,\textbf{x}_*), \cdots, k(\bm{\mu}_K,\textbf{x}_*)]^T$ and $\textbf{K} = [k(\bm{\mu}_i,\bm{\mu}_j)]_{\bm{\mu}_i,\bm{\mu}_j \in \mathcal{K}}$.
Using the above equations we can compute the GP to update our knowledge about the unknown function $f$ based on information acquired through observations. 

Note that the variance computed using Equation \ref{eq:sigma^2} does not depend of the observed values $y_i$ but only on the locations $\bm{\mu}_i$ of the training set. This is an important property of GPs and plays a significant role in our contribution.

The predictive performance of GPs depends exclusively on the suitability of the chosen kernel and parameters.
There are lots of possible kernel functions to choose from \cite{Rasmussen,Murphy:2012}. A common property is to have the covariance decrease as the distance between the points grows, so that the prediction is mostly based on the near locations. 
A famous and widely use kernel is the squared exponential, also known as Gaussian kernel:

\begin{equation} \label{eq: SEK}
k(\bm{a},\bm{b}) = \sigma_f^2 \exp\bigg( - \frac{(\bm{a}-\bm{b})^T(\bm{a}-\bm{b})}{2l^2} \bigg)
\end{equation}

The kernel function will often have some parameters, for example, a length parameter that determines how quickly the covariance decreases with distance. In the squared exponential kernel $l$ controls the horizontal length scale over which the function varies, and $\sigma_f^2$  controls the vertical variation.

\subsection{Submodular functions} \label{subsec:Submodular functions}

A set function is a function which takes as input a set of elements. Particular classes of set functions turn out to be submodular. A fairly intuitive characterization of a submodular function has been given \cite{Nemhauser1978}: 
\begin{definition}
A function $F$ is submodular if and only if for all $A \subseteq B \subseteq X$ and $x \in X \setminus B$ it holds that: 
\begin{equation} \label{eq: submodular}
F(A \cup \{x\})-F(A) \geq F(B \cup \{x\})-F(B)
\end{equation}
\end{definition}

This definition captures a concept known as diminishing return property.
Informally we can say that if $F$ is submodular, adding a new element $x$ to a set increases the value of $F$ more if we have fewer elements than if we have more.
This property allows for a simple greedy forward-selection algorithm with an optimality bound guarantees \cite{Nemhauser1978}. This is widely exploited in sensing optimization literature \cite{Krause-sensor1,Krause07,Powers,krause2014submodular,tzoumas2018resilient}. 

This concept is of our interest as it directly apply to Gaussian processes.
Specifically, the posterior variance of a Gaussian process belongs to this class of submodular functions. \cite{das2008algorithms} show that the variance reduction:
\begin{equation}
F_{\textbf{x}}(A)= \sigma^2(\textbf{x})-\sigma^2(\textbf{x}|A)
\end{equation}
at any particular point $\textbf{x}$, satisfies the diminishing returns property: adding a new observation reduces the variance in $\textbf{x}$ more if we have made few observations so far, and less if we have already made many observations.
%A submodular-based algorithm would optimize Equation \ref{eq:F(A)} by greedily selecting the elements of in the set $A$.

\subsection{Quadratic Unconstrained Binary Optimization (QUBO)} \label{subsec: Quadratic Unconstrained Binary Optimization (QUBO)}

The goal of a Quadratic Unconstrained Binary Optimization problem is to find the assignment of a set of binary variables $z_1 ... z_n$ so as to minimize a given objective function: 
\begin{equation} \label{eq:QUBOObjectiveFunction}
O(z_1, ..., z_n)= \sum \limits_{i=1}^n a_i z_i + \sum \limits_{1 \leq i < j \leq n} b_{i,j} z_i z_j
\end{equation}

Each instance of a QUBO problem can be conveniently represented by using a weighted graph where each node $i$ represents a binary variable $z_i$, a linear coefficient $a_i$ encodes the value associated to the node $i$ and a quadratic coefficient $b_{i,j}$ encodes the value associated to the edge between nodes $i$ and $j$.

In this graphical representation the QUBO objective function \eqref{eq:QUBOObjectiveFunction} corresponds to the summation of the values in the graph, namely the sum of linear terms will be the sum of the node values and the sum of the quadratic terms will be the sum of the edge values: 
% \begin{equation} \label{eq:QUBO graph objective function}
% 	O(z_1, ..., z_n)= \underbrace{\sum \limits_{i=1}^n a_i z_i}_\text{node values} +  \underbrace{\sum \limits_{1 \leq i < j \leq n} b_{i,j} z_i z_j}_\text{edge values} 
% \end{equation}
Hence, the minimization of this objective function is equivalent to decide which nodes to remove in such a way that the summation of values remaining in the graph is the lowest possible. Notice that the removal of a node implies the removal of all edges that are incident to that node.

% \subsubsection{Example}
% The function $O(z_1,z_2,z_3,z_4,z_5)=3z_1 - z_2 + 3z_3 + z_5 + z_1z_2 - z_2z_3 - z_1z_4 - z_1z_5 - 2z_4z_5 + 2z_3z_5$ can be represented as:
% \begin{center}
% 	\begin{tikzpicture}
% 		[scale=1,auto=left,every node/.style={circle,draw}]
% 		\node (n1) at (0,0) [label=above:$z_1$] {3};
% 		\node (n2) at (4,0) [label=above:$z_2$] {-1};
% 		\node (n3) at (8,-1.5)[label=right:$z_3$] {3};
% 		\node (n4) at (0,-3)[label=below:$z_4$] {0};
% 		\node (n5) at (4,-3)[label=below:$z_5$] {1};
% 		\draw (n1) -- (n2) node [midway, draw=none, fill=none] {1};
% 		\draw (n2) -- (n3) node [midway, draw=none, fill=none] {-1};
% 		\draw (n1) -- (n4) node [midway, draw=none, fill=none] {-1};
% 		\draw (n1) -- (n5) node [midway, draw=none, fill=none] {-1};
% 		\draw (n4) -- (n5) node [midway, draw=none, fill=none] {-2};
% 		\draw (n3) -- (n5) node [midway, draw=none, fill=none] {2};
% 	\end{tikzpicture}
% \end{center}
% and minimizes for $O(0,1,0,1,1) = -2$ represented as:
% \begin{center}
% 	\begin{tikzpicture}
% 		[scale=1,auto=left,every node/.style={circle,draw}]
% 		\node (n2) at (4,0) [label=above:$z_2$] {-1};
% 		\node (n4) at (0,-3)[label=below:$z_4$] {0};
% 		\node (n5) at (4,-3)[label=below:$z_5$] {1};
% 		\draw (n4) -- (n5) node [midway, draw=none, fill=none] {-2};
% 	\end{tikzpicture}
% \end{center}

\subsection{Problem definition} \label{subsec:Problem definition}

Given a Gaussian process and a discretized domain $\mathcal{X}$, we want to select a set of $K$ points within $\mathcal{X}$ where to perform measurements in order to minimize the total posterior variance of the Gaussian Process. Specifically we want to select a set of $K$ measurements taken at locations $\mathcal{K}=\{\textbf{x}_1,\textbf{x}_2,\cdots,\textbf{x}_K\}$ such that we minimize the following objective function: 
\begin{equation} \label{eq: OBJ(chapQUBOGD)}
J(\{\textbf{x}_1,\textbf{x}_2,\cdots,\textbf{x}_K\}) = \sum_{\textbf{x}_i \in \mathcal{X}} \sigma^2(\textbf{x}_i)
\end{equation}
where $\sigma^2(x)$ is the Gaussian process variance computed as defined by Equation \ref{eq:sigma^2}.  
Notice that the variance computed with Equation \ref{eq:sigma^2} in locations $\textbf{x}_i \in \mathcal{X}$ is dependent on the set of sampling points $\mathcal{K}$ as explained in Section \ref{subsec:Gaussian Processes}.

Our goal is to model a QUBO objective function (Equation \ref{eq:QUBOObjectiveFunction}) that approximate our problem:
\begin{equation} \label{eq:approximation}
O(z_1, \dots , z_n) \approx J(\{\textbf{x}_1,\textbf{x}_2,\cdots,\textbf{x}_K\})
\end{equation}

\section{QUBO model for GP variance reduction} \label{sec:QUBO model for GP variance reduction}

Given a domain $\mathcal{X}$ the QUBO model will be a complete graph composed of $|\mathcal{X}|$ nodes. Specifically, it's a graph where every node is connected with every other node. Hence, the total number of edges is $\frac{|\mathcal{X}|(|\mathcal{X}|-1)}{2}$.
In the following description we explain how the values of the graph are set. We decompose the study of the model into two parts:
\begin{enumerate}
\item How we set the values related to the variance of the Gaussian Process.
\item How we set the values in order to implement the constraint of the problem (i.e, the number of sampling points that must be exactly $K$).
\end{enumerate}

\subsection{Variance values} \label{subsec:Variance values}

As previously mentioned given a domain $\mathcal{X}$, that is the set of candidate sampling locations, we build a graph composed by $|\mathcal{X}|$ nodes where every node corresponds to a location of the domain.
Since our problem requires to select a set of sampling locations such that the variance of the Gaussian process is minimized, a natural representation is to assign to each node of the graph the amount of variance reduction that is obtainable by sampling the GP in the location represented by that node. That is, the value of node $i$ is: 

\begin{equation} \label{eq:alpha_i}
\alpha_i \triangleq J(\{\textbf{x}_i\})-J(\emptyset )
\end{equation}

Given the properties of Gaussian processes this definition of $\alpha_i$ guarantees negative values.
Since the objective function of a QUBO instance has to be minimized and the objective of our problem is to minimize the variance of a Gaussian Process, Equation \ref{eq:alpha_i} sets $\alpha_i$ values as the negative amount of variance reduction. Selecting a node gives us an improvement (i.e, a lower value of the QUBO objective function) equivalent to the amount of variance that we can reduce from the GP by selecting the sampling location represented by this node of the graph. 
We can imagine these $\alpha_i$ as \dquotes{selecting forces}, in fact the amount of variance reduction obtainable is a force that requires a specific node to be selected as a candidate sampling location.

When we select two sampling locations the amount of variance reduction of the Gaussian process is lower than the sum of the reduction obtainable from the two sampling locations alone, specifically:
\begin{equation} \label{eq:varianceproperty}
J(\emptyset)-J(\{\textbf{x}_i,\textbf{x}_j\}) < \big(J(\emptyset) - J(\{\textbf{x}_i\})\big) + \big(J(\emptyset)-J(\{\textbf{x}_j\})\big)
\end{equation}

As a consequence, in our QUBO objective function we have to keep into account the \squotes{mutual information} between sampling locations. To this aim we set the edges values in our graph as follows:

\begin{align} \label{eq:beta_ij}
\beta_{i,j} &\triangleq J(\{\textbf{x}_i, \textbf{x}_j\})-J(\emptyset ) - \alpha_i - \alpha_j \nonumber \\
		&= J(\{\textbf{x}_i, \textbf{x}_j\})-J(\emptyset ) -J(\{\textbf{x}_i\})+J(\emptyset ) -J(\{\textbf{x}_j\})+J(\emptyset ) \nonumber \\
        &= J(\{\textbf{x}_i, \textbf{x}_j\}) -J(\{\textbf{x}_i\}) -J(\{\textbf{x}_j\}) +J(\emptyset ) 
\end{align}

We can imagine these $\beta_{i,j}$ almost as the opposite as a \dquotes{selecting force}. The variance reduction obtainable is proportional to the distance between points in space and location close to each other would have a high mutual information. As a consequence the value $\beta_{i,j}$ between these locations would be high discouraging the simultaneous selection of these two point in space. Please note that what makes two points near or far depends on the length-scale of the process encoded in the hyperparameters of the kernel used.

The values $\beta_{i,j}$ are guaranteed to be always positive as we prove in the following theorem.

\begin{theorem} \label{thm:valuesbij}
Values $\beta_{i,j}$ are guaranteed to be higher that 0.
\end{theorem}
\begin{proof}
Given Equation \ref{eq:varianceproperty}:
\begin{align} 
&J(\emptyset)-J(\{\textbf{x}_i,\textbf{x}_j\}) < \big(J(\emptyset) - J(\{\textbf{x}_i\})\big) + \big(J(\emptyset)-J(\{\textbf{x}_j\})\big) \nonumber \\
&J(\emptyset)-J(\{\textbf{x}_i,\textbf{x}_j\}) < 2J(\emptyset) - J(\{\textbf{x}_i\}) - J(\{\textbf{x}_j\}) \nonumber \\
&-J(\emptyset)-J(\{\textbf{x}_i,\textbf{x}_j\}) + J(\{\textbf{x}_i\}) + J(\{\textbf{x}_j\}) < 0 \nonumber \\
&J(\emptyset)+J(\{\textbf{x}_i,\textbf{x}_j\}) - J(\{\textbf{x}_i\}) - J(\{\textbf{x}_j\}) > 0 \nonumber
\end{align}
\end{proof}

Given the current model, the objective function is minimized when we select a number of nodes (sampling points) that is dependent on the values on the graph which, in turn, depend on the domain and kernel of the Gaussian process.
Since the problem asks to select a specific given number $K$ of sampling points, it is now important to implement a constraint that guarantees that the correct number is selected. This is described in the next section.

\subsection{Implementation of the constraint} \label{subsec:Implementation of the constraint}

In order to implement a constraint into an unconstrained problem we have to encode it as penalties in our objective function. In case of a QUBO model that has to be minimized, any combination that does not satisfy the constraint must have an higher value that prevents the selection of such configuration as an optimal solution.

In what follows we describe how we can implement the constraint of our problem into a complete graph. To do so, we further divide the description of our implementation in two separate parts, specifically: 
\begin{enumerate}
\item How to guarantee that exactly $K$ nodes are selected in an zero-value graph (i.e. a weighted graph with zero values on every node and every edge).
\item How to guarantee that the constraint is strong enough given that we want to implement it in a non zero-value graph.
\end{enumerate}

\subsubsection{Selecting $K$ nodes}

Here, we describe how we can set the values in a complete graph such that only $K$ nodes are selected.
In order to do so we have to guarantee that the objective function is minimized if and only if exactly $K$ nodes are selected. For any other cases the objective function needs to have a higher value.

Starting from a zero-value graph, to implement the constraint we need to add some values on the nodes and edges. Specifically, we will assign the same value $A \in \mathbb{R}$ to every node and the same value $B \in \mathbb{R}$ to every edge.
If in a complete graph we select $n$ nodes, this lead to the following sum of values:
\begin{equation} \label{eq:general sum}
\frac{Bn(n-1)}{2} + An
\end{equation}

If we want to guarantee that exactly $K$ nodes are selected Equation \ref{eq:general sum} needs to have its minimum value when $n=K$.

\begin{theorem} \label{thm:constraint values}
Given the function $\frac{Bn(n-1)}{2} + An$, for any $B>0$ if $A = -BK+\frac{B}{2}$ the minimum is in $n=K$.
\end{theorem}
\begin{proof}
\begin{equation}
\dfrac{Bn(n-1)}{2} + An = \dfrac{Bn^2}{2} - \dfrac{Bn}{2} + An \nonumber
\end{equation}
With $B>0$ the quadratic function is a parabola that opens upwards and the minimum of the function is where the derivative is equal to $0$.

\begin{align}
\dfrac{\partial}{\partial n} \Big( \dfrac{Bn^2}{2} - \dfrac{Bn}{2} + An \Big) = 0 \nonumber\\
Bn - \dfrac{B}{2} + A = 0 \nonumber
\end{align}

Now we want to fix the derivative to be $0$ when $n=K$.

\begin{align}
BK - \dfrac{B}{2} + A = 0 \nonumber \\
A = -BK + \dfrac{B}{2} \nonumber
\end{align}

\end{proof}

With theorem \ref{thm:constraint values} we have shown that it is possible to implement the constraint for any $K$, hence also in the restricted case of our problem when $K \in \mathbb{N}^{[2,|X|-1]}$.

\subsubsection{Guarantee the strength of the constraint}

In the discussion above we have shown how to implement the constraint in a zero-value graph. However, in our case we have to guarantee that it is satisfied in a complete graph that is already populated with values. In this section we show that values $A$ and $B$ can be set such that the energy penalties are strong enough to guarantee that the constraint is always satisfied.

Given $A = -BK+\frac{B}{2}$, Equation \ref{eq:general sum} becomes:

\begin{equation}
\dfrac{Bn(n-1)}{2} - BKn +\dfrac{Bn}{2} = \dfrac{Bn^2}{2} -BKn
\end{equation}

We want to analyze how this function increases as we move away from $n=K$ that represent the minimum. Specifically, given $l \in \mathbb{N}$ we analyze how this function increases when $n=K+l$ and when $n=K-l$.

\begin{align} \label{eq:constraint strength piu}
\Big[ \dfrac{Bn^2}{2} -BKn \Big]_{n=K+l} - \Big[ \dfrac{Bn^2}{2} -BKn \Big]_{n=K}   &=  \nonumber\\
\dfrac{B(K+l)^2}{2}-BK(K+l)- \Big( \dfrac{BK^2}{2} - BK^2 \Big) &= \dfrac{Bl^2}{2} 
\end{align}

\begin{align} \label{eq:constraint strength meno}
\Big[ \dfrac{Bn^2}{2} -BKn \Big]_{n=K-l} - \Big[ \dfrac{Bn^2}{2} -BKn \Big]_{n=K}   &=  \nonumber\\
\dfrac{B(K-l)^2}{2}-BK(K-l)- \Big( \dfrac{BK^2}{2} - BK^2 \Big) &= \dfrac{Bl^2}{2} 
\end{align}

With Equations \ref{eq:constraint strength piu} and \ref{eq:constraint strength meno} we show that the constraint act as an energy penalty by increasing the value of the objective function quadratically with the difference of number of nodes selected with respect to a given $K$.

Note that the constraint that we have built is generic and can be applied to any other problems where a specific given number of nodes has to be selected. From the mathematical point of view of a QUBO function the constraint built so far guarantees that exactly $K$ binary variables needs to have value $1$ in order to minimize the function.
Moreover, any feasible combination (combinations where $K$ binary variables has value $1$) starts with the same energy value as expressed in Equation \ref{eq:baseEnergy}, leaving these solutions to `compete' for which  is the optimal one for a specific instance.

\begin{equation} \label{eq:baseEnergy}
\dfrac{BK(K-1)}{2} -BK^2 + \dfrac{BK}{2} = -\dfrac{BK^2}{2}
\end{equation}

The strength of the constraint is directly dependent on the value $B$. In order to guarantee that it is satisfied in a non zero-value graph (i.e. a weighted graph with non-zero weights) we have to selected the value of $B$ big enough to overcome additional `forces'. 

From a practical point of view, since $B \in \mathbb{R}^+$ we can just set it as a very big number to guarantee that the constraint is satisfied. For our specific problem we show in the next section that it is possible to compute a bound for the value $B$.

\subsection{Ensuring a lower bound for the constraint}

Here we want to compute a lower bound for the value of $B$ such that the constraint is satisfied in a generic weighted graph with values computed as presented in Section \ref{subsec:Variance values}. Let start by making the following considerations: 
\begin{itemize}
\item The constraint is satisfied if and only if the minimum of the QUBO objective function correspond to a solution with exactly $K$ nodes selected.
\item The strength of the constraint (i.e. the energy penalty) as computed in Equations \ref{eq:constraint strength piu} and \ref{eq:constraint strength meno} is $\dfrac{Bl^2}{2}$, where $l\in \mathbb{N}$ represents the difference from $K$ on the number of selected points.
\item Values $\alpha_i$ computed with Equation \ref{eq:alpha_i} are always negative.
\item Values $\beta_{i,j}$ computed with Equation \ref{eq:beta_ij} are always positive.
\end{itemize}

Given these considerations, we can compute what is the strongest \dquotes{force} that a non feasible solution is applying to deviate from a feasible solution. In other terms, what is the highest contribution to the QUBO function that a configuration without the constraint satisfied is applying. Once we compute this value, we can set $B$ to be high enough such that the energy penalty overcome the worst case scenario. In what follow we analyze the two possible cases.

\paragraph{Configuration with more nodes selected}~\\
Let $\mathcal{A}$ be the set of the $\alpha_i$ values computed with Equation \ref{eq:alpha_i}, that is, the set of values assigned to the nodes and that represent the amount of variance reduction obtainable by sampling in that point.

Now, we define a set $\mathcal{A}_n$ as the set of the $n$ lowest values of the set $\mathcal{A}$:

\begin{equation}
\mathcal{A}_n \triangleq \begin{cases}
\mathcal{A}_0 = \emptyset \\
\mathcal{A}_n = \mathcal{A}_{n-1} \cup \min(\mathcal{A} \setminus \mathcal{A}_{n-1})
\end{cases}
\end{equation}

Since we want to minimize the QUBO objective function the contribution of values $\alpha_i$ correspond to a \dquotes{force} that is trying to add more nodes to the solution, whereas values $\beta_{i,j}$ on the other side opposes to the selection of more nodes. For a configuration with $K+l$ nodes selected, the strongest contribution of forces that is trying to deviate from a feasible solution of $K$ nodes is the following: 
\begin{equation}
\underbrace{\sum_{\alpha_i \in \mathcal{A}_l} |\alpha_i|}_\text{Contribution of additional nodes} - \underbrace{\bigg( \dfrac{(K+l)(K+l-1)}{2} - \dfrac{K(K-1)}{2} \bigg)\beta_{i,j}}_\text{Contribution of additional edges}
\end{equation}

We can now compute an upper bound for this quantity as follows: 
\begin{align} \label{eq: forza piu nodi}
\sum_{\alpha_i \in \mathcal{A}_l} |\alpha_i| - \bigg( \dfrac{(K+l)(K+l-1)}{2} - \dfrac{K(K-1)}{2} \bigg)\beta_{i,j} &< \nonumber \\
\sum_{\alpha_i \in \mathcal{A}_l} |\alpha_i| - 0 &\leq \nonumber \\
l |\min(\mathcal{A})| &\leq \nonumber \\
l^2 |\min(\mathcal{A})| &
\end{align}

Now, if we want the constraint to be strong enough to overcome any configuration with more than $K$ nodes selected we need to impose the strength of the constraint computed in Equations \ref{eq:constraint strength piu} and \ref{eq:constraint strength meno} to be higher then the upper bounded force applied by non-feasible configuration as computed in Equation \ref{eq: forza piu nodi}, that is:

\begin{align} \label{eq: B per piu nodi}
\dfrac{Bl^2}{2} &> l^2 |\min(\mathcal{A})| \nonumber \\
B & > 2 |\min(\mathcal{A})|
\end{align}

\paragraph{Configuration with less nodes selected}~\\
Similarly to what we have done above, here we compute a bound for the value $B$ such that the energy penalty of the constraint is strong enough to overcome configurations with less than $K$ nodes selected. 

Let $\mathcal{B}$ be the set of the $\beta_{i,j}$ computed with Equation \ref{eq:beta_ij}, that is, the set of the value that we have assigned for the moment to the edges.
We define a set $\mathcal{B}_n$ as the set of the $n$ highest values of the set $\mathcal{B}$:
\begin{eqnarray}
\mathcal{B}_n \triangleq \begin{cases}
\mathcal{B}_0 = \emptyset \\
\mathcal{B}_n = \mathcal{B}_{n-1} \cup \max (\mathcal{B} \setminus \mathcal{B}_{n-1})
\end{cases}
\end{eqnarray}

Since we want to minimize the QUBO objective function the contribution of values $\beta_{i,j}$ correspond to a \dquotes{force} that is trying to remove nodes to the solution, whereas values $\alpha_i$ on the other side opposes to the removal of nodes. For a configuration with $K-l$ nodes selected, the strongest contribution of forces that is trying to deviate from the feasible solution of $K$ nodes is the following:

\begin{equation}
\underbrace{\sum_{ \beta_{i,j} \in \mathcal{B}_{\big(\frac{K(K-1)}{2} - \frac{(K-l)(K-l-1)}{2}\big)}} \beta_{i,j}}_\text{Contribution of removing edges} \;- \underbrace{l|\alpha_i|}_\text{contribution of removing nodes}
\end{equation}

where $\frac{K(K-1)}{2} - \frac{(K-l)(K-l-1)}{2}$ represents the number of extra edges if we select $K$ sensors as opposed to $K-l$.
We can now compute an upper bound for this quantity as follow:

\begin{align} \label{eq: forza meno nodi}
\sum_{ \beta_{i,j} \in \mathcal{B}_{(\frac{K(K-1)}{2} - \frac{(K-l)(K-l-1)}{2})}} \beta_{i,j} - l|\alpha_i| &< \nonumber \\
\sum_{ \beta_{i,j} \in \mathcal{B}_{(\frac{K(K-1)}{2} - \frac{(K-l)(K-l-1)}{2})}} \beta_{i,j} - 0 \leq \nonumber \\
\bigg(\dfrac{K(K-1)}{2} - \dfrac{(K-l)(K-l-1)}{2}\bigg) \max (\mathcal{B}) &= \nonumber \\
\bigg( Kl - \dfrac{l^2}{2} - \dfrac{l}{2} \bigg) \max (\mathcal{B}) &= \nonumber \\
\dfrac{l^2}{2}\bigg( \dfrac{2K}{l} -1 -\dfrac{1}{l} \bigg) \max (\mathcal{B}) &< \nonumber \\
\dfrac{2Kl^2}{2} \max (\mathcal{B})
\end{align}

Now, if we want the constraint to be strong enough to overcome any configuration with less than $K$ nodes selected we need to impose the strength of the constraint computed in Equations \ref{eq:constraint strength piu} and \ref{eq:constraint strength meno} to be higher then the upper bounded force applied by non-feasible configuration as computed in Equation \ref{eq: B per meno nodi}, that is:

\begin{align} \label{eq: B per meno nodi}
\dfrac{Bl^2}{2} &> \dfrac{2Kl^2}{2} \max (\mathcal{B}) \nonumber \\
B &> 2K \max (\mathcal{B})
\end{align}

Now, given the nature of the problem and considering that we want to present a general model that works with any possible hyperparameters combinations of the Gaussian process, we cannot infer which one between Equations \ref{eq: B per piu nodi} and \ref{eq: B per meno nodi} imposes a bigger value of $B$. However, for any instance of the QUBO objective function we can easily compute a feasible value for $B$ as follows: 

\begin{equation} \label{eq: finalB}
B > \max \Big(  2 |\min(\mathcal{A})|   \,,\,   2K \max (\mathcal{B}) \Big)
\end{equation}

\subsection{The complete model}

The complete model can easily be expressed as a complete weighted graph whose values are the sum of the values as seen in section \ref{subsec:Variance values} and the constraint that we explained above.

Specifically, the nodes of the graph will have value:
\begin{equation} \label{eq:a_i totale} 
a_i \triangleq \alpha_i -BK +\dfrac{B}{2}
\end{equation}

The edges of the graph will have value:
\begin{equation} \label{eq:b_ij totale}
b_{i,j} \triangleq \beta_{i,j} + B
\end{equation}

To conclude, the final QUBO instance will be:
\begin{align} \label{eq: QUBO complete model}
O(z_1, ..., z_n)= \sum \limits_{i=1}^{|\mathcal{X}|} \Big(J(\{x_i\})-J(\emptyset ) -BK +\dfrac{B}{2}\Big) z_i + \nonumber \\
\sum \limits_{1 \leq i < j \leq |X|} \Big(J(\{x_i, x_j\}) -J(\{x_i\}) -J(\{x_j\}) +J(\emptyset )  + B\Big) z_i z_j
\end{align}
with B computed as in Equation \ref{eq: finalB}.

\subsection{Optimized variant} \label{subsec:Optimized variant}

Here we propose a variant of the model described above. 
Notice that values $\beta_{i,j}$ computed in Equation \ref{eq:beta_ij} represent the difference between the variance reduction obtainable with two sampling points acquired at the same time and the sum of the variance reduction obtainable with the two points acquired one at a time.

However, $\beta_{i,j}$ is not taking into account how the variance in the Gaussian process is affected by the presence of other measurement locations. 
In general, the variance of a Gaussian process in monotonically decreasing with the number of sampling points that we add. 
Hence to better represent the real variance reduction of the Gaussian process in case we have more than 2 sampling locations, the $\beta_{i,j}$ values should be lower then what computed by Equation \ref{eq:beta_ij}.

As a consequence, we propose a variant where we multiply by a weight $w$ values $\beta_{i,j}$ computed by Equation \ref{eq:beta_ij}.
This $w$ multiplier is intended as a scaling factor for the difference in the variance reduction, hence $\beta_{i,j}$ is computed as follows: 

\begin{align} \label{eq:beta_ij con w}
\beta_{i,j} &\triangleq w\big(J(\{\textbf{x}_i, \textbf{x}_j\})-J(\emptyset ) - \alpha_i - \alpha_j \big)\nonumber \\
&= w\big(J(\{\textbf{x}_i, \textbf{x}_j\}) -J(\{\textbf{x}_i\}) -J(\{\textbf{x}_j\}) +J(\emptyset ) \big)
\end{align}

By setting $0<w\leq 1$ we can better approximate the real variance values of the Gaussian process for cases where we have 3 or more sampling locations.
We will show in the empirical evaluation (Section \ref{sec: Empirical Evaluation(QUBO GP)}) that this additional parameter allows us to obtain better results.

\section{Empirical Evaluation} \label{sec: Empirical Evaluation(QUBO GP)}

% Specifically, in contrast to our previous contributions that optimizes in a continuous manner, here the possible sensing points can be selected from a set of predetermined available locations, namely, the domain points $\mathcal{X}$ of which the domain of interest is discretized.
% In this context our direct competitors are discrete (combinatorial) techniques such as the greedy submodular optimization \cite{Krause-sensor1,Krause07,Powers} (see Section \ref{sec: Optimizing sampling locations} for related work in this context).

In what follows we present the results of the empirical evaluation of our QUBO model for Gaussian process posterior variance reduction. 
Notice that we are not proposing an optimization method for quadratic unconstrained binary problems, instead we want to show that the QUBO objective function represented by our model is a good approximation of the problem as explained in Section \ref{subsec:Problem definition}.
The main objectives of this empirical evaluation are:
\begin{enumerate}
\item Show a comparison between the optimal solution of the the QUBO model with respect to submodular optimization.
\item Show a comparison with a random sampling solution used as a simple baseline technique.
\item Test the model under different conditions:
\begin{itemize}
\item Using different hyperparameters of the Gaussian process to show the generality of the approach.
\item Showing how the model behaves when varying the number of sampling points $K$.
\end{itemize}
\end{enumerate}

\subsection{Dataset and setup} \label{subsec:Dataset and setup}

We have generated two 2-dimensional cubic datasets with equally distributed domain points $\mathcal{X}$.
Specifically, the cardinalities of the domains $|\mathcal{X}|$, that is the number of points on which we evaluate the Gaussian process, are 25 and 36 points respectively.
We tested the procedure by training the Gaussian process using the squared exponential kernel reported in Equation \ref{eq: SEK}.

In our tests we used two different length-scale $l$, two different $\sigma_f$ and two different $\sigma_n$. As previously mentioned in Section \ref{subsec:Gaussian Processes}, $l$ describes the smoothness property of the true underlying function, $\sigma_f$ describes the standard deviation of the modeled function and $\sigma_n$ the standard deviation of the noise of the observation. These hyperparameters gives us a total o 8 different combinations to test the model under different conditions.

For the variant presented in Section \ref{subsec:Optimized variant} we have tested our model using 19 different parameter $w$, specifically from 0.1 up to 1 with steps of 0.05. Notice that using $w=1$ corresponds to the the basic version of the model described by Equation \ref{eq: QUBO complete model}.
To compute the optimum value of the QUBO objective function under all these different settings we used CPLEX optimization library.
All the above mentioned combination of hyperparameters and datasets have been tested by adapting a different number $K$ of measurement points which varies from 2 up to 7. The case of a single point has been excluded since the submodular greedy technique is optimal by definition.

Moreover, regarding the comparison with a random procedure, for each of the described combination of hyperparameters and number $K$ of sampling locations, we have generated 100 randomly selected solutions (i.e. randomly selected sampling locations between the domain points $\mathcal{X}$.

\subsection{Results} \label{subsec:Results(QUBO_GP)}

Figures \ref{fig:results(QUBO_GP)25p} and \ref{fig:results(QUBO_GP)36p} show the aggregated results of the experiments previously described. In these plots we can observe the total remaining variance of the Gaussian process (in logarithmic scale for the sake of representation) by varying the number $K$ of sampling points. Each line in these charts represents the average over the eight combination of hyperparameters used in the experiment. Moreover for random technique the line represents the average over $8\times100$ experiments (8 combination of hyperparameters and 100 randomly selected combinations of sampling points).

\begin{figure}[ht]
	\centering
    \includegraphics[keepaspectratio, width=1\linewidth]{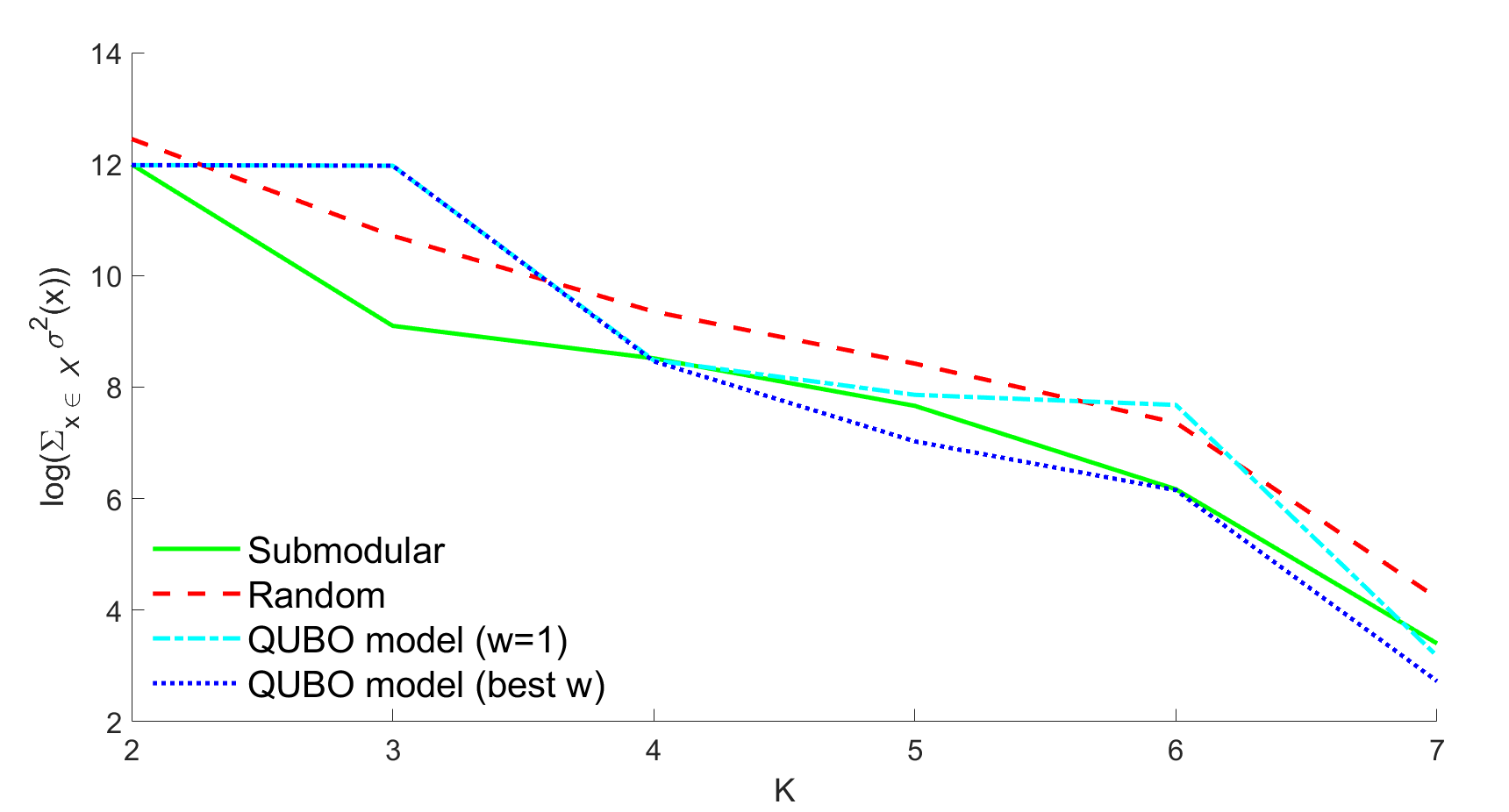}
    \caption{Remaining variance of the Gaussian process by varying the number $K$ of sampling locations on the dataset with a domain composed of 25 points. The results represent the average over the 8 combination of hyperparameters used during the experiment. }
    \label{fig:results(QUBO_GP)25p}
\end{figure}

\begin{figure}[ht]
	\centering
    \includegraphics[keepaspectratio, width=1\linewidth]{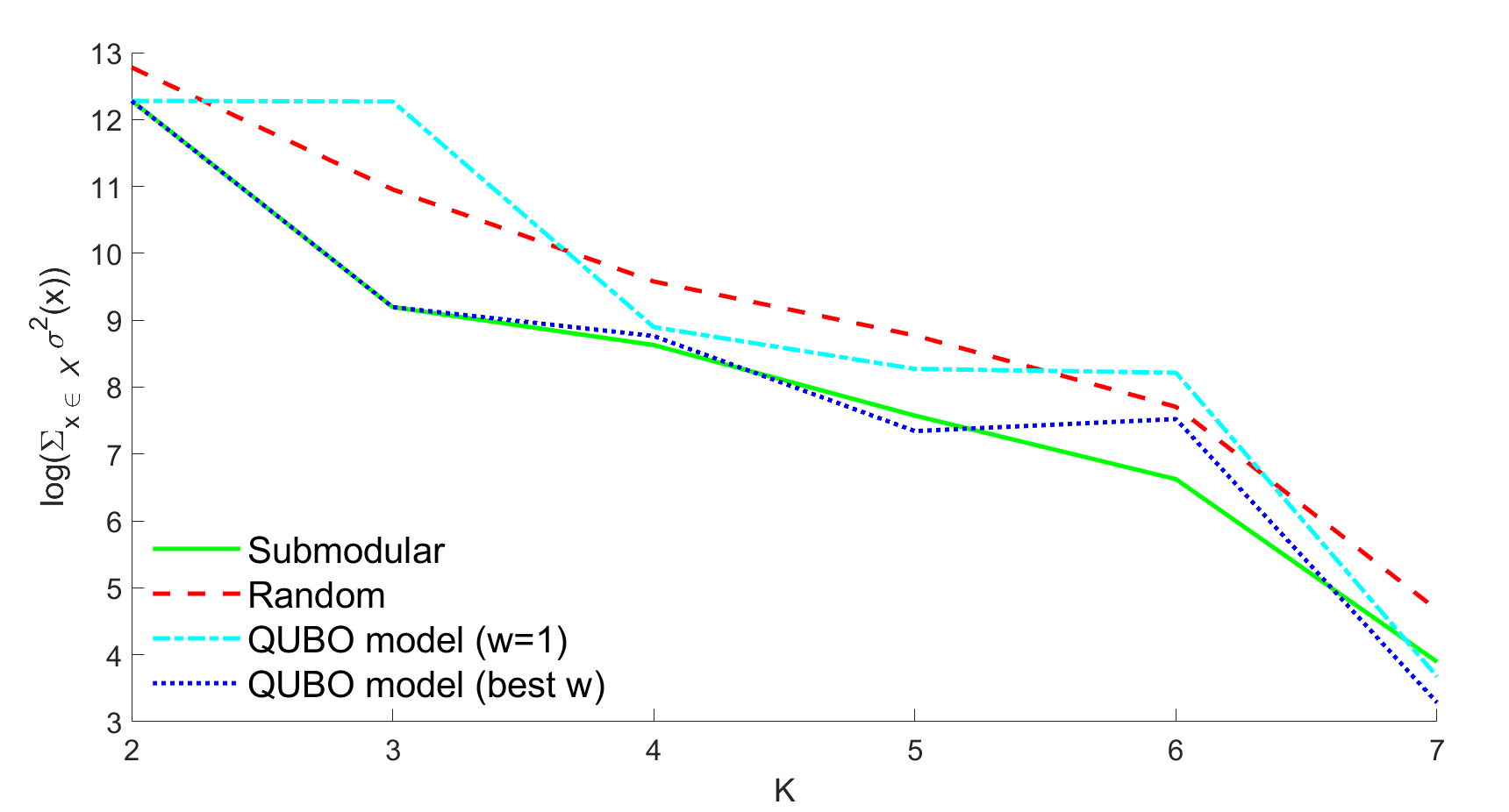}
    \caption{Posterior variance of the Gaussian process by varying the number $K$ of sampling locations on the dataset with a domain composed of 36 points. The results represent the average over the 8 combination of hyperparameters used during the experiment. }
    \label{fig:results(QUBO_GP)36p}
\end{figure}

First of all we notice that in general the optimized variant of the QUBO model (described in Section \ref{subsec:Optimized variant}) where we module the strength of the quadratic terms by a parameter $w$, allows us to obtain much better results compared to the \squotes{standard} QUBO model (that correspond to use $w=1$), proving that indeed a good tuning of that parameter provides an advantage for our QUBO model. 
Regarding the $w$ parameter we remand to the consideration presented in the next section.
%Hence, in the following considerations we will take into account this variant. 

Moreover, we can observe that the optimum solutions of our QUBO model (by tuning the $w$ parameter) in the first dataset (Figure \ref{fig:results(QUBO_GP)25p}) are comparable with submodular selection technique when using 2 sampling points, worst with 3 and better in all the other cases. A similar situation is true for the second dataset (Figure \ref{fig:results(QUBO_GP)36p}) where we have in some cases a comparable results, in some worst and in the remaining  better results than submodular.

On average the solutions obtained with a random selection of sampling points perform worst than both submodular and the QUBO model.

As we can observe for both the datasets (Figures \ref{fig:results(QUBO_GP)25p} and \ref{fig:results(QUBO_GP)36p}) the trend for all the techniques tested is the same. As expected, by adding more measurement locations the variance of the Gaussian process decreases. However, the interleaving of the curve representing our QUBO model and the curve representing the submodular selection shows that the optimum of the QUBO model represent indeed a good approximation of the objective function that we are approximating.

\section{Conclusions} \label{sec: Conclusions}

In this paper we proposed a novel QUBO model to tackle the problem of optimizing sampling locations in order to minimize the posterior variance of a Gaussian process. 
The strength of this contribution is the proposal of a completely alternative method that can be used by non-classical computing architectures (quantum annealer) and therefore benefit from research in this field.

Although the $w$ parameter of our model has to be determined empirically, results shows that the optimum of the QUBO objective function represent a good solution for the above mentioned problem, obtaining comparable and in some cases better results than the widely used submodular technique.

%Moreover, values $\beta_{i,j}$ as computed in Equation \ref{eq:beta_ij} could be replaced with values computed with different techniques.
We believe that our contribution with this QUBO model takes an important first step towards the sampling optimization of Gaussian processes in the context of quantum computation.

\bibliographystyle{unsrt}
\bibliography{biblio}

\end{document}